\documentclass[review]{elsarticle}
\usepackage[left=1.7cm,top=2.5cm,right=1.7cm,bottom=2.5cm]{geometry}
\usepackage{lineno,hyperref}
\usepackage{graphicx,array,delarray}
\usepackage{subfigure}
\usepackage{latexsym}
\usepackage{amscd}
\usepackage{amsfonts, amsmath, amssymb}
\usepackage[utf8]{inputenc}
\usepackage{verbatim}
\usepackage{mathrsfs}
\usepackage{fancyhdr}
\usepackage{float}
\usepackage{palatino}
\modulolinenumbers[5]
\newtheorem{theorem}{\textbf{Theorem}}

\newtheorem{corollary}[theorem]{\textbf{Corollry}}

\newtheorem{definition}[theorem]{\textbf{Definition}}

\newtheorem{lemma}[theorem]{\textbf{Lemma}}

\newtheorem{proposition}[theorem]{\textbf{Proposition}}
\newtheorem{remark}[theorem]{Remark}

\newenvironment{proof}[1][Proof]{\noindent\textbf{#1.} }{\ \rule{0.5em}{0.5em}}

\journal{Journal of The Franklin Institute}

\begin{document}

\begin{frontmatter}

\title{\textbf{Matrix-valued Laurent polynomials, parametric linear systems and integrable systems}}

\author[add1]{Nancy L\'{o}pez-Reyes\corref{mycorrespondingauthor}}
\cortext[mycorrespondingauthor]{Corresponding author}
\ead{nancy.lopez@udea.edu.co}
\author[add2]{Raul Felipe-Sosa}
\author[add3]{Raul Felipe}
%%%%%%%%%%%%%%%%%%%%%%%%%%%%%%%%%%

\address[add1]{Instituto de Matem\'{a}ticas, Universidad de Antioquia \\
Calle 67 No. 53 - 108. Medellin, Colombia}
\address[add2]{FCFM, Benem\'{e}rita Universidad Aut\'{o}noma de Puebla \\
4 Sur 104 Centro Hist\'{o}rico C.P. 72000. Puebla, M\'{e}xico.}
\address[add3]{CIMAT \\
         Callej\'on Jalisco s/n Mineral de Valenciana \\
         Guanajuato, Gto, M\'exico.\\}

\begin{abstract}
In this paper, we study transfer functions corresponding to parametric linear systems whose coefficients are block matrices. Thus, these transfer functions constitute Laurent polynomials whose coefficients are square matrices. We assume that block matrices defining the parametric linear systems are solutions of an integrable hierarchy called by us, the block matrices version of the finite discrete KP hierarchy, which is introduced and studied with certain detail in this paper. We see that the linear system defined of the simplest solution of the integrable system is controllable and observable. Then, as a consequence of this fact, it is possible to verify that any solution of the integrable hierarchy, obtained by the dressing method of the simplest solution, defines a parametric linear system which is  also controllable and observable.
\end{abstract}


\end{frontmatter}

\section{Introduction}

The study of parametric linear systems has been developed from the works of Brockett and Krishnaprasad \cite{brockett}.
The realization theory affirms that each linear system has a unique rational function associated with it. Through this
correspondence it has been studied some identification problems for these linear systems, see for instance \cite{brockett}, \cite{felipe2}, \cite{nancy}.
The evolution of the coefficients with respect to the parameters leads to the respective evolution of rational functions (called usually transfer functions).
Several authors have studied parametric linear systems for which the main coefficient is a solution with respect
to the parameters of certain integrable systems (hierarchies). We explain how they arise, if we fix a rational function then it can be
written by means of the simplest solution of the integrable system, and we also have an initial linear system for which the rational function is
its transfer function. Taking into account that
any integrable system is always related to some type of group factorization one can construct a family of linear systems having
as main coefficient the solutions of the integrable system. This approach to build families of parametric linear systems leads to an interesting
relationship between the linear control theory and the integrable systems. Indeed, a central question is which properties, from the point
of view of the linear control theory, are inherited from the initial linear system for the remaining elements of the family. The present
research is devoted to this question. Among the previous works, we must mention the article by Y. Nakamura \cite{nakamura} where the
Toda lattice is used. In the Nakamura paper, the reader can also consult other important references. Previous work by some of the authors
can be found in \cite{felipe2}, \cite{nancy}, \cite{nancy1}.

During the 70s and 80s of the last century, the $(2+1)$-dimensional KP-equation was studied in detail in the framework of the theory of integrable systems,
both in the continuous case and in the discrete one. It constitutes one of the key equations of mathematical physics. Among the reasons for this is its
relationship with the so-called traveling waves and solitons.
By this same time, Date, Jimbo, Kashiwara y Jimbo (more exactly in $1981$) introduced the $B$-type KP-equation or $BKP$-equation. This equation arises
from the $B$-type Lie algebras as opposed to the usual $KP$-equation which arises from the $A$-type algebras. For both nonlinear equations, special solutions can be obtained through their Hirota form (or bilinear form), such is the case of the $N$-lump solutions. In recent years, again the B-type KP-equation has resumed his prominence. In particular, recently the BKP-equation (also adding the case $(3+1)$-dimensional) has been linked to the study of rogue waves and other related topics (the interested reader can consult the following books \cite{kharif} and \cite{onorato}). Notably, these types of waves do not appear only in fluids, they have been observed in other media. An excellent mathematical reference is provided by the following articles which we strongly recommend \cite{fenga}, \cite{peng}, \cite{qin}, \cite{tian}, \cite{tu}, \cite{wang}, \cite{wang1}, \cite{wang2} and \cite{yan}.

Recently M. C. C\^{a}mara, A. F. dos Santos and P. F. dos Santos \cite{camara} have considered matrix equations of Lax type of the following form
\begin{equation}\label{i1}
\frac{d N(t,z)}{dt}=[N^{+}(t,z),N(t,z)],
\end{equation}
where the $n\times n$ matrix $N(t,z)$ depends of a parameter $z$ called spectral parameter varying on the unit circle $S^{1}$ ($N^{+}(t,z)$ is constructed through $N(t,z)$). Specifically,
$N(t,z)$ is a matrix-valued Laurent polynomial in $z$ and $N^{+}(t,z)$ is the part of $N(t,z)$ analytic in the unit disc $\mathbb{D}$.

Denote by $\left [ C^{1}(I) \right ]^{n\times n}$ the space of continuously differentiable $n\times n$ matrix functions on the open interval $I\subset \mathbb{R}^{+}$ (with respect to the variable $t$),
where $I$ will be considered a neighborhood of the origin.
The authors of the above mentioned work \cite{camara} considered the equation (\ref{i1}) with respect to Laurent polynomials $N(t,z)\in \left [ C^{1}(I) \right ]^{n\times n}$ of the form
\begin{equation}\label{i2}
N(t,z)=\sum_{k=-m}^{1}P_{k}(t)z^{k}=P_{1}(t)z+P_{0}(t)+\frac{P_{-1}(t)z^{m-1}+\cdots+P_{-m}(t)}{z^{m}} \quad\quad(m\in\mathbb{N},z\in S^{1}),
\end{equation}
for which in (\ref{i1}) we have $N^{+}(t)=P_{1}(y)z+P_{0}(t)$.

\begin{remark}In this moment, we must clarify the equation (\ref{i1}) for $N(t,z)$ given by (\ref{i2}), in particular we indicate how to understand $[N^{+}(t,z),N(t,z)]$.
Let $V$ an associative algebra over $\mathbb{C}$ and $z\neq 0$ a parameter not necessarily confined to $S^{1}$. Let $V((z^{-1}))=
\{\sum_{i=-\infty}^{q}v_{i}z^{i}|v_{i}\in V,\,\,q\in \mathbb{Z}\}$ be the set of all formal Laurent series with coefficients in $V$
and $V[z,z^{-1}]=\{\sum_{i=p}^{q}v_{i}z^{i}|v_{i}\in V,\,\,p,q\in \mathbb{Z}\}$ its subset whose elements are
formal Laurent polynomials. For any $L=\sum_{i=-\infty}^{q}v_{i}z^{i}$ we set $L_{+}=\sum_{i\geq 0}^{q}v_{i}z^{i}$, $L_{-}=\sum_{i< 0}^{q}v_{i}z^{i}$.

The elements in $V[z,z^{-1}]$ of the form $vz^{i}$ for $v\in V$ are called monomials. If we have two monomials $vz^{i}$ and $wz^{j}$ we can define the
Lie bracket or Lie product of
these monomials as the monomial $[vz^{i},wz^{j}]=[v,w]z^{i+j}$, where as usual for associative algebras $[v,w]=vw-wv$. This product can be extended by linearity to $V[z,z^{-1}]$.
In our case $V=M_{n}(\mathbb{C})$ is the algebra of complex matrices of order $n$. Then, the equation (\ref{i1}) must be understood as an equality between two elements of $M_{n}(\mathbb{C})[z,z^{-1}]$. We would like to observe that by taking some suitable subsets of $M_{n}(\mathbb{C})$ and restrict the parameter $z$ to $S^{1}$, in the study of (\ref{i1}) we can meet in a natural way with the notions of Loop algebras, Riemann-Hilbert problem etc
\cite{semenov}. However, it is not our goal to discuss these topics in this paper.
Now, since
\begin{equation*}
\frac{\partial N(t,z)}{\partial t}=\frac{\partial P_{1}(t)}{\partial t}z+\frac{\partial P_{0}(t)}{\partial t}+\frac{\partial P_{-1}(t)}{\partial t}z^{-1}
+\cdots +\frac{\partial P_{-m}(t)}{\partial t}z^{-m},
\end{equation*}
and
\begin{align*}
[N^{+}(t,z),N(t,z)]&=[P_{1},P_{-1}]+\left ([P_{1},P_{-2}]+[P_{0},P_{-1}] \right )z^{-1}+\left ([P_{1},P_{-3}]+[P_{0},P_{-2}] \right )z^{-2}+\cdots + \\
&\,\,\,\,\,\,\left ([P_{1},P_{-(k+1)}]+[P_{0},P_{-k}] \right )z^{-k}+\cdots +\left ([P_{1},P_{-m}]+[P_{0},P_{-(m-1)}] \right )z^{-(m-1)}+[P_{0},P_{-m}]z^{-m},
\end{align*}
then we conclude that the matrix $P_{1}(t)$ must be constant and (\ref{i1}) is equivalent to the nonlinear system
\begin{equation*}
\frac{\partial P_{0}(t)}{\partial t}=[P_{1},P_{-1}],\cdots ,\frac{\partial P_{-k}(t)}{\partial t}=[P_{1},P_{-(k+1)}]+[P_{0},P_{-k}],\cdots,
\frac{\partial P_{-m}(t)}{\partial t}=[P_{0},P_{-m}],
\end{equation*}
where $1\leq k \leq m-1$.
\end{remark}

In this paper, unlike the work mentioned above, we consider the matrix-valued Laurent polynomial as the transfer function of certain linear system and for this system, we study its properties of controllability and observability when the coefficients evolve by means of an integrable hierarchy. Justly, we consider the particular case of Laurent polynomial (\ref{i2}) with $P_{0}=P_{1}=O_{n}$, where $O_{k}$ stands for the $k\times k$ zero matrix for any $k\geq 1$. This is the fundamental reason why below we do not use the Lax equation (\ref{i1}), instead we introduce and study a block matrix version of the finite discrete KP hierarchy.
As we already mentioned, in the present work, we only will consider matrix-valued Laurent polynomials of the form
\begin{equation}\label{ri2}
L(t,z)=\sum_{k=-m}^{1}P_{k}(t)z^{k}=\frac{P_{-1}(t)z^{m-1}+\cdots+P_{-m}(t)}{z^{m}}.
\end{equation}

From now on, we will assume a more general situation in which $L=L(t_{1},\cdots, t_{m-1})\in [C^{1}(I^{m-1})]^{n\times n}$, in other words, any matrix $P_{-k}$ involved in the definition of our Laurent polynomial $L$ depends of $m-1$ variables $t_{1},\cdots, t_{m-1}$ for $2\leq m$, and each one of these variables takes values in $I$. Besides, in this work, unless otherwise specified, all matrices will have real entries.

Next, we briefly review the $(k-1)$-dimensional left-projective spaces over the real or complex $n\times n$ matrices \cite{schwarz-zaks}.
Real or complex $tn\times sn$ matrices with $t,s\geq 1$ and $t\neq s$ or $t=s$ for $t,s\geq 2$ are denoted by calligraphic capital letters. One writes the $n\times sn$ matrix $\mathcal{Y}$ in block form: $\mathcal{Y}=(Y_{1},\cdots,Y_{s})$, in which each $Y_{i}$ is an $n\times n$ matrix. $R_{0}(sn^{2})$ will be the set of real or complex $n\times sn$ matrices $\mathcal{Y}$ of rank equal to $n$, that is,
the $n$ rows of $\mathcal{Y}$ are linearly independent (this, in turn is identical to the maximal number of linearly independent columns of $\mathcal{Y}$). In other words, $\mathcal{Y}\in R_{0}(sn^{2})$
if and only if by definition the dimension of the vector space spanned by its $n$ rows is exactly $n$ (it shows that $n$ is also the dimension of the vector space spanned by its columns). The notation here
is very important, for instance, let's pay attention a moment to $R_{0}(32)$ for which the parameters $s$ and $n$ can take different values. If $\mathcal{Y}\in R_{0}(32)$ in the case for which $s=2$ and $n=4$
then one can find $4$ columns of $\mathcal{Y}$ representing a basis of $\mathbb{C}^{4}$, while if $\mathcal{Y}\in R_{0}(32)$ when $s=8$ and $n=2$ then this implies that between the $16$ columns of $\mathcal{Y}$
there exists at least a basis of $\mathbb{C}^{2}$. This tells us that it is convenient to keep the notation $R_{0}(2(4)^{2})$, $R_{0}(8(2)^{2})$, etc instead of $R_{0}(32)$.

$R_{0}(sn^{2})$ is a connected topological space and its topology is defined by means of any matrix norm.

Two matrices $\mathcal{Y}=(Y_{1},\cdots,Y_{s})$ and $\mathcal{U}=(U_{1},\cdots,U_{s})$ of $R_{0}(sn^{2})$ are left- or row-equivalent if there exists an $n\times n$ invertible matrix $S$ such that
\begin{equation}\label{s1}
\mathcal{U}=(U_{1},\cdots,U_{s})=(SY_{1},\cdots,SY_{s})=S\mathcal{Y},\,\,\,\,\,\,|S|\neq 0.
\end{equation}

This relation partitions $R_{0}(sn^{2})$ into equivalence classes of row-equivalent matrices. These equivalence classes are the points of the $(s-1)$-dimensional left-projective space over the real or complex $n\times n$ matrices $\mathbb{P}_{(s-1)}(M_{n}(\mathbb{K}))$, where $\mathbb{K}$ is $\mathbb{R}$ or $\mathbb{C}$. The projective mappings $\mathcal{C}$ of this left-projective space are given by means of constant invertible $sn\times sn$ matrices. $\mathcal{C}$ is written in block form
\begin{equation}\label{s2}
\mathcal{C}=\left(
              \begin{array}{ccc}
                C_{11} & \cdots & C_{1s} \\
                \vdots &  & \vdots \\
                C_{s1} & \cdots & C_{ss} \\
              \end{array}
            \right),\,\,\,\,\,\,\,\,|\mathcal{C}|\neq 0,
\end{equation}
where each block $C_{ij}$, $i,j=1,\cdots,s$ is an $n\times n$ matrix. For $\mathcal{C}$ fixed, one defines
\begin{equation}\label{s3}
\widetilde{\mathcal{Y}}=(\widetilde{Y}_{1},\cdots,\widetilde{Y}_{s})=\mathcal{C}(\mathcal{Y})=\mathcal{Y}\mathcal{C}
=(Y_{1},\cdots,Y_{s})\left(
              \begin{array}{ccc}
                C_{11} & \cdots & C_{1s} \\
                \vdots &  & \vdots \\
                C_{s1} & \cdots & C_{ss} \\
              \end{array}
            \right),
\end{equation}
for all $\mathcal{Y}\in \mathbb{P}_{(s-1)}(M_{n}(\mathbb{K}))$, then $\mathcal{C}(\mathcal{Y})\in \mathbb{P}_{(s-1)}(M_{n}(\mathbb{K}))$.
If $\mathcal{U}=S\mathcal{Y}$ where $|S|\neq 0$, then $\widetilde{\mathcal{U}}=\mathcal{U}\mathcal{C}=S\mathcal{Y}\mathcal{C}=S\widetilde{\mathcal{Y}}$;
hence row-equivalent matrices have row-equivalent transformations. Thus, the transformation (\ref{s3}) induces a transformation of
$\mathbb{P}_{(s-1)}(M_{n}(\mathbb{K}))$ onto itself. From now on, for our purpose, it could be convenient to use matrices of $R_{0}(mn^{2})$ and
invertible $mn\times mn$ matrices which will be written in block form.

We would like to continue this section with an observation on $L(t,z)$ given by (\ref{ri2}) which represents an extension of the theory of
realization to matrix-valued Laurent polynomials of the form (\ref{ri2}). We have
\begin{equation}\label{i3}
L(t,z)=\left (I_{n},O_{n},\ldots,O_{n} \right )\Pi(z)(P^{T}_{-1}(t),\ldots,P^{T}_{-m}(t))^{T},
\end{equation}
where $\left (I_{n},O_{n},\ldots,O_{n} \right )\in R_{0}(mn^{2})$ and $I_{k}$ denotes for $1\leq k$ the identity matrix of order $k$, moreover

\begin{align*}
\Pi(z)=(zI_{nm}-\Lambda)^{-1}
&=\left ( z \left(
                           \begin{array}{ccccc}
                             I_{n} & O_{n} & \cdots & \cdots & O_{n} \\
                             O_{n} & \ddots & \ddots & \ddots & \vdots \\
                             \vdots & \ddots & \ddots & \ddots & \vdots \\
                             \vdots & \ddots & \ddots & \ddots & O_{n} \\
                             O_{n} & \cdots & \cdots & O_{n} & I_{n} \\
                           \end{array}
                         \right)
                         -\left(
                           \begin{array}{ccccc}
                             O_{n} & I_{n} & O_{n} & \cdots & O_{n} \\
                             \vdots & \ddots & \ddots & \ddots & \vdots \\
                             \vdots & \ddots & \ddots & \ddots & O_{n} \\
                             \vdots & \ddots & \ddots & \ddots & I_{n} \\
                             O_{n} & \cdots & \cdots & \cdots & O_{n} \\
                           \end{array}
                         \right)\right )^{-1} \\
\,\,\,\,\,\,\,\,&=\left(
                           \begin{array}{ccccc}
                             \frac{I_{n}}{z} & \frac{I_{n}}{z^{2}} & \frac{I_{n}}{z^{3}} & \cdots & \frac{I_{n}}{z^{m}} \\
                             O_{n} & \ddots & \ddots & \ddots & \vdots \\
                             \vdots & \ddots & \ddots & \ddots & \frac{I_{n}}{z^{3}} \\
                             \vdots & \ddots & \ddots & \ddots & \frac{I_{n}}{z^{2}} \\
                             O_{n} & \cdots & \cdots & O_{n} & \frac{I_{n}}{z} \\
                           \end{array}
                         \right),
\end{align*}
where $\Lambda$ is the shift block matrix of order $mn\times mn$.

The equality (\ref{i3}) holds for all $t\in I$, in particular
\begin{equation}\label{i6}
L(0)=\left (I_{n},O_{n},\ldots,O_{n} \right )\Pi(z) (P^{T}_{-1}(0),\ldots,P^{T}_{-m}(0))^{T}.
\end{equation}

We recall the following result which is known as the \textbf{Schur determinant lemma} (see \cite{zhang} for more details)

\begin{lemma}Let $P,Q,S,R$ denote $n\times n$ matrices and suppose that $P$ and $R$ commute. Then the determinant $|M|$ of the $2n\times 2n$
matrix
\begin{equation*}
M=\left(
    \begin{array}{cc}
      P & Q \\
      R & S \\
    \end{array}
  \right),
\end{equation*}
is equal to the determinant of the matrix $PS-RQ$.
\end{lemma}

There exists a generalization in certain sense of the previous result which can be found also in \cite{zhang} for any square matrix $M$. Consider now that $M$ is partitioned where $P,Q,S,R$
do not necessarily have the same dimension. Suppose $P$ is nonsingular and denote the matrix $S-RP^{-1}Q$ by $M/P$ and call it the Schur complement of $P$ in $M$, or the Schur complement of $M$
relative to $P$. In the same spirit, if $S$ is nonsingular, the Schur complement of $S$ in $M$ is $M/S=P-QS^{-1}R$. The following result is well known
\begin{theorem}(Schur's Formula) Let $M$ be a partitioned square matrix. If $P$ is nonsingular, then
\begin{equation}\label{i7}
det(M/P)=\frac{det M}{det P}.
\end{equation}
\end{theorem}

\section{Definition of the hierarchy}

In this section, we present the bases that allow us to build and study our integrable hierarchy. More exactly, we introduce a block
matrix version of the finite discrete KP hierarchy through the $mn\times mn$ block matrix shift operator $\Lambda$ acting
on $mn\times n$ column matrices $\mathcal{Y}=(Y_{1},\cdots,Y_{m})^{T}$ where each $Y_{k}$ is an $n\times n$ matrix for any $k$, that is
\begin{equation}\label{B1}
\Lambda=\left(
                           \begin{array}{ccccc}
                             O_{n} & I_{n} & O_{n} & \cdots & O_{n} \\
                             \vdots & \ddots & \ddots & \ddots & \vdots \\
                             \vdots & \ddots & \ddots & \ddots & O_{n} \\
                             \vdots & \ddots & \ddots & \ddots & I_{n} \\
                             O_{n} & \cdots & \cdots & \cdots & O_{n} \\
                           \end{array}
                         \right),
\end{equation}
in particular, we develop a block matrix Borel-Gauss approach for this integrable system.

Below for two square matrices $A$ and $B$ of the same order, we use the notation $[A,B]$ to indicate the Lie product of both matrices,
that is, $[A,B]=AB-BA$. Define
\begin{equation}\label{B2}
H=\Lambda+D_{0}+\sum_{k=1}^{m-1}D_{k}\left( \Lambda^{T}\right )^{k},
\end{equation}
where the $D_{k}$ are $mn\times mn$ block diagonal matrices for $k=0, 1\ldots, m-1$. The entries of $H$ are assumed to be functions of $m-1$ variables $t_{1},\ldots, t_{m-1}$. The $mn\times mn$ matrix $H$ will be called a Lax matrix if it satisfies the following equations
\begin{equation}\label{B3}
\frac{\partial H}{\partial t_{r}}=\left [H^{r}_{\geq},H\right ],\quad\quad\quad r=1,\ldots,m-1,
\end{equation}
where $M_{\geq}$ ($M_{>}$) denotes the (strictly) upper triangular part of a matrix $M$, analogously $M_{\leq}$ ($M_{<}$) denotes
the (strictly) lower triangular part of $M$.
The set of equations (\ref{B3}) is called the block matrix finite discrete KP hierarchy.
Observe that the simplest solution of the hierarchy (\ref{B3}) is $H=\Lambda$.

Next, we will clarify these equations but first we will make some simple observations of linear algebra related with the shift block matrix $\Lambda$\,:
\begin{itemize}
  \item Note that $\Lambda\Lambda^{T}$ and $\Lambda\Lambda^{T}$ are both diagonal matrices. Indeed, we have
  \begin{equation*}
  \Lambda\Lambda^{T}=\left(
                       \begin{array}{ccccc}
                         I_{n} & O_{n} & \ldots & \ldots & O_{n} \\
                         O_{n} & \ddots & \ddots & \ddots & \vdots \\
                         \vdots & \ddots & \ddots & \ddots & \vdots \\
                         \vdots & \ddots & \ddots & I_{n} & \vdots \\
                         O_{n} & \ldots & \ldots & \ldots & O_{n} \\
                       \end{array}
                     \right),\,\,\,\,\,\,\,\,\,\,\,\,\,\,\Lambda\Lambda^{T}=\left(
                       \begin{array}{ccccc}
                         O_{n} & \ldots & \ldots & \ldots & O_{n} \\
                         \vdots & I_{n} & \ddots & \ddots & \vdots \\
                         \vdots & \ddots & \ddots & \ddots & \vdots \\
                         \vdots & \ddots & \ddots & \ddots & O_{n} \\
                         O_{n} & \ldots & \ldots & O_{n} & I_{n} \\
                       \end{array}
                     \right).
  \end{equation*}
  \item Let $D$ be a $nm\times nm$ diagonal block matrix, then there is a $mn\times mn$ diagonal block matrix $R$ such that
  $(\Lambda^{T})^{k}D=R(\Lambda^{T})^{k}$ where $1\leq k\leq m-1$. Specifically,
  \begin{equation*}
  R=(O_{n},\cdots,O_{n},(D)_{11},\cdots,(D)_{(m-k)(m-k)}),
  \end{equation*}
  see the proof of proposition \ref{propreview1} for more details.
  \item In our paper, we take advantage of the fact that $(\Lambda^{T})^{m}=(\Lambda^{T})^{m+1}=(\Lambda^{T})^{m+2}=\cdots=O_{nm}$, that is,
  $\Lambda^{T}$ is a nilpotent matrix of degree equal to $m$.
  \item Let $A=B(\Lambda^{T})^{k}$ where $B$ is a diagonal block matrix. Then, the block matrix $A$ has all its entries equal to zero, except
  those of its k-th diagonal below the main diagonal. Additionally, it shows that $A$ is a strict lower triangular matrix. It implies that
  if $P^{1}=D^{1}_{1}\Lambda^{T}+D^{1}_{2}(\Lambda^{T})^{2}+\cdots +D^{1}_{m-1}(\Lambda^{T})^{m-1}$, and $P^{2}=D^{2}_{1}\Lambda^{T}+D^{2}_{2}(\Lambda^{T})^{2}+\cdots +D^{2}_{m-1}(\Lambda^{T})^{m-1}$,
  where each $D^{j}_{i}$ is a diagonal block matrix for $j=1,2$ and $1\leq i \leq m-1$ are two strict lower triangular matrices, then $P^{1}=P^{2}$ if and only if $D^{1}_{k}=D^{2}_{k}$ for $k=1,\cdots, m-1$.
  Of course, the previous observation also applies in the case where the diagonal block matrices $D^{j}_{i}$ are functions of a variable $t$. Again, if $A$ is an arbitrary diagonal block $nm\times nm$
  matrix (in particular $A$ will be the identity $I_{nm}$) then in order to $A+P^{1}=A+P^{2}$ is necessary and sufficient that $D^{1}_{k}=D^{2}_{k}$ for all $k$.
  \item If $A$ is an arbitrary diagonal block matrix, then $\Lambda A(\Lambda)^{T}$ is a diagonal block matrix. More exactly,
  \begin{equation*}
  \Lambda A(\Lambda)^{T}=((A)_{22},\cdots,(A)_{mm},O_{n}).
  \end{equation*}
\end{itemize}

We return to our goal of clarifying the equations (\ref{B3}). Claim that each side of the equations (\ref{B3}) for a fixed $r$ is of the form
\begin{equation*}
L_{r}=L_{0}^{r}+L_{1}^{r}\Lambda^{T}+\cdots+L_{k}^{r}(\Lambda^{T})^{k}+\cdots+L_{m-1}^{r}(\Lambda^{T})^{m-1},
\end{equation*}
where any $L_{k}^{r}$ is a diagonal block matrix for $1\leq k\leq m-1$. In the case of the left side of (\ref{B3}) is easy to prove the affirmation,
because
\begin{equation*}
L_{r}^{left}=\frac{\partial H}{\partial t_{r}}=\frac{\partial D_{0}}{\partial t_{r}}+\sum_{k=1}^{m-1}\frac{\partial D_{k}}{\partial t_{r}}
\left( \Lambda^{T}\right )^{k}.
\end{equation*}

The presentation of $L_{r}^{right}$ requires a little more attention and it makes use of some of the previous observations. First notice that
for a fixed $r$
\begin{align*}
[H^{r}_{\geq}, H]&=[H^{r}_{\geq}+H^{r}_{<}-H^{r}_{<}, H]=[H^{r}-H^{r}_{<}, H]=-[H^{r}_{<}, H]=[H,H^{r}_{<}]=[\Lambda+D_{0}+H_{<},H^{r}_{<}] \\
&=[\Lambda+D_{0},H^{r}_{<}]+[H_{<},H^{r}_{<}],
\end{align*}
now, it is clear from the previous observations that $[D_{0}+ H_{<},H^{r}_{<}]=L^{right,1}_{1}\Lambda^{T}+\cdots+L^{right,1}_{m-1}(\Lambda^{T})^{m-1}$,
where the $L^{right,1}_{k}$ are diagonal block matrices of order $nm$ for $1\leq k\leq m-1$. On the other hand, $[\Lambda,H^{r}_{<}]=L^{right,2}_{0}+L^{right,2}_{1}\Lambda^{T}+\cdots+L^{right,2}_{m-1}(\Lambda^{T})^{m-1}$, it implies that
$[H^{r}_{\geq}, H]=L^{right}_{0}+L^{right}_{1}\Lambda^{T}+\cdots+L^{right}_{m-1}(\Lambda^{T})^{m-1}$, where $L^{right}_{k}$ is a diagonal
block matrix of order $nm$ for any $k$ with $0\leq k \leq m-1$, because $L^{right}_{0}=L^{right,2}_{0}$ and $L^{right}_{k}=L^{right,1}_{k}+L^{right,2}_{k}$
for $1\leq k \leq m-1$. Hence, the equation (\ref{B3}) for this $r$ fixed is equivalent to the following closed nonlinear system of ordinary
differential equations in the entries of $H$\,:
\begin{equation*}
\frac{\partial D_{0}}{\partial t_{r}}=L^{right}_{0},\cdots\cdots \frac{\partial D_{k}}{\partial t_{r}}=L^{right}_{k},\cdots\cdots
\frac{\partial D_{m-1}}{\partial t_{r}}=L^{right}_{m-1},
\end{equation*}
here $0\leq k \leq m-1$.

Let us consider $mn\times mn$ block matrices of the form
\begin{equation}\label{B4}
S=I_{nm}+\sum_{k=1}^{m-1}S_{k}\left( \Lambda^{T}\right )^{k},
\end{equation}
where $S_{k}$ is in the class of all $mn\times mn$ block diagonal matrices for $k=1,\ldots,m-1$.

One can see that a matrix of the form (\ref{B4}) is not singular, that is, there exists $S^{-1}$. Indeed, suppose that $Sx=0$ for
$x=(\overline{x}_{1},\cdots,\overline{x}_{m})\in \mathbb{R}^{mn}$. Then, we have
\begin{align*}
\left(
  \begin{array}{c}
    O_{n} \\
    \vdots \\
    \vdots \\
    \vdots \\
    O_{n} \\
  \end{array}
\right)&=\left(
  \begin{array}{c}
    \overline{x}_{1} \\
    \vdots \\
    \vdots \\
    \vdots \\
    \overline{x}_{m} \\
  \end{array}
\right)+\left(
                           \begin{array}{ccccc}
                             D^{1}_{11} & O_{n} & \cdots & \cdots & O_{n} \\
                             O_{n} & \ddots & \ddots & \ddots & \vdots \\
                             \vdots & \ddots & \ddots & \ddots & \vdots \\
                             \vdots & \ddots & \ddots & \ddots & O_{n} \\
                             O_{n} & \cdots & \cdots & O_{n} & D^{1}_{mm} \\
                           \end{array}
                         \right)\left(
                           \begin{array}{ccccc}
                             O_{n} & O_{n} & \cdots & \cdots & O_{n} \\
                             I_{n} & \ddots & \ddots & \ddots & \vdots \\
                             \vdots & \ddots & \ddots & \ddots & \vdots \\
                             \vdots & \ddots & \ddots & \ddots & O_{n} \\
                             O_{n} & \cdots & \cdots & I_{n} & O_{n} \\
                           \end{array}
                         \right)\left(
  \begin{array}{c}
    \overline{x}_{1} \\
    \vdots \\
    \vdots \\
    \vdots \\
    \overline{x}_{m} \\
  \end{array}
\right)+\cdots  \\
&\,\,\,\,\,\,\,\,+\left(
                           \begin{array}{ccccc}
                             D^{m-1}_{11} & O_{n} & \cdots & \cdots & O_{n} \\
                             O_{n} & \ddots & \ddots & \ddots & \vdots \\
                             \vdots & \ddots & \ddots & \ddots & \vdots \\
                             \vdots & \ddots & \ddots & \ddots & O_{n} \\
                             O_{n} & \cdots & \cdots & O_{n} & D^{m-1}_{mm} \\
                           \end{array}
                         \right)\left(
                           \begin{array}{ccccc}
                             O_{n} & O_{n} & \cdots & \cdots & O_{n} \\
                             O_{n} & \ddots & \ddots & \ddots & \vdots \\
                             \vdots & \ddots & \ddots & \ddots & \vdots \\
                             O_{n} & \ddots & \ddots & \ddots & O_{n} \\
                             I_{n} & O_{n} & \cdots & O_{n} & O_{n} \\
                           \end{array}
                         \right)\left(
  \begin{array}{c}
    \overline{x}_{1} \\
    \vdots \\
    \vdots \\
    \vdots \\
    \overline{x}_{m} \\
  \end{array}
\right),
\end{align*}
now, from this equality follows that
\begin{equation*}
\begin{array}{c}
  \overline{x}_{1}=O_{n}, \\
  \overline{x}_{2}+D^{1}_{22}\overline{x}_{1}=O_{n}, \\
  \overline{x}_{3}+D^{1}_{33}\overline{x}_{2}+D^{2}_{33}\overline{x}_{1}=O_{n}, \\
  \cdots \\
  \cdots \\
  \overline{x}_{m-1}+D^{1}_{(m-1)(m-1)}\overline{x}_{m-2}+\cdots+D^{m-3}_{(m-1)(m-1)}\overline{x}_{2}+D^{m-2}_{(m-1)(m-1)}\overline{x}_{1}=O_{n}, \\
   \overline{x}_{m}+D^{1}_{mm}\overline{x}_{m-1}+D^{2}_{mm}\overline{x}_{m-2}+\cdots+D^{m-2}_{mm}\overline{x}_{2}+D^{m-1}_{mm}\overline{x}_{1}=O_{n}.
\end{array}
\end{equation*}

This implies that $\overline{x}_{1}=\overline{x}_{2}=\cdots=\overline{x}_{m-1}=\overline{x}_{m}=O_{n}$. Hence, $x=0$. For instance, for $m=2$
we have that $(\mathcal{I}+S_{1}\Lambda^{T})^{-1}=\mathcal{I}-S_{1}\Lambda^{T}$.
In this sense, we have

\begin{proposition}\label{propreview1}
Let $m\leq 2$ be fixed and the matrix $S$ define by (\ref{B4}) given, then its inverse matrix $S^{-1}$ is of the form
\begin{equation}\label{B5}
S^{-1}=I_{nm}+\sum_{k=1}^{m-1}S^{I}_{k}\left( \Lambda^{T}\right )^{k},
\end{equation}
where for each $k$ the matrix $S^{I}_{k}$ is in the class of all $mn\times mn$ block diagonal matrices.
\end{proposition}
\begin{proof}
Since the usual product of matrices is associative our proof is based on two simple results which are true in any associative algebra:
$1)$ the inverse of an element of the algebra if it exists is unique, $2)$ the inverse on the left of an element, is also the inverse on the right.
Thus, we propose the inverse on the left of (\ref{B4}) in the form (\ref{B5}) and proceed to calculate the coefficients of (\ref{B5}) in a recurring way.
In other words, we suggest to calculate the coefficients $S^{I}_{k}$ of $S^{-1}$ for any $k$ of the following equation
\begin{equation}\label{equationinv}
\left(I_{nm}+S^{I}_{1}\Lambda^{T}+S^{I}_{2}(\Lambda^{T})^{2}+\cdots+S^{I}_{m-1}(\Lambda^{T})^{m-1}\right)
\left(I_{nm}+S_{1}\Lambda^{T}+S_{2}(\Lambda^{T})^{2}+\cdots+S_{m-1}(\Lambda^{T})^{m-1}\right)=I_{nm},
\end{equation}
then first we find $S^{I}_{1}$, after we calculate $S^{I}_{2}$ etc. Concretely it follows that $S^{I}_{1}=-S_{1}$, immediately after the coefficient $S^{I}_{2}$ is obtained from equality
$$S^{I}_{2}(\Lambda^{T})^{2}+S^{I}_{1}\Lambda^{T}S_{1}\Lambda^{T}+S_{2}(\Lambda^{T})^{2}=O_{nm}(\Lambda^{T})^{2},$$
where as before $O_{nm}$ is the $nm\times nm$ null matrix. Indeed, $\Lambda^{T}S_{1}=D\Lambda^{T}$ where $D=Diag (O_{n},(S_{1})_{11},\ldots,(S_{1})_{(m-1)(m-1)})$, hence $S^{I}_{2}=S_{1}D-S_{2}$.
In a similar way, we can calculate $S^{I}_{3}$ from the equation
\begin{equation*}
S^{I}_{3}(\Lambda^{T})^{3}+S^{I}_{2}(\Lambda^{T})^{2}S_{1}\Lambda^{T}+S^{I}_{1}\Lambda^{T}S_{2}(\Lambda^{T})^{2}+S_{3}(\Lambda^{T})^{3}
=O_{nm}(\Lambda^{T})^{3},
\end{equation*}
it shows that $S^{I}_{3}=-(S^{I}_{2}A+S^{I}_{1}B+S_{3})$, where
$$A=Diag (O_{n},O_{n},(S_{1})_{11},\ldots,(S_{1})_{(m-2)(m-2)}),$$ and
$$B=Diag (O_{n},(S_{2})_{11},\ldots,(S_{2})_{(m-1)(m-1)}),$$
hence, $S^{I}_{3}=-((S_{1}D-S_{2})A-S_{1}B+S_{3})=S_{1}(B-D)+S_{2}A-S_{3}$.

Explicitly, we have obtained $S_{1}^{I}$, $S_{2}^{I}$ and $S_{3}^{I}$. We continue the proof by  induction in $k$ where $k< m-1$. Suppose that
we have already calculated $S_{1}^{I}, S_{2}^{I}, S_{3}^{I},\ldots, S_{k}^{I}$ then from (\ref{equationinv}) it follows
\begin{align*}
&S_{k+1}^{I}(\Lambda^{T})^{k+1}+S_{k}^{I}(\Lambda^{T})^{k}S_{1}\Lambda^{T}+S_{k-1}^{I}(\Lambda^{T})^{k-1}S_{2}(\Lambda^{T})^{2}+\cdots+
S_{2}^{I}(\Lambda^{T})^{2}S_{k-1}(\Lambda^{T})^{k-1} \\
&+S_{1}^{I}\Lambda^{T}S_{k}(\Lambda^{T})^{k}+S_{k+1}(\Lambda^{T})^{k+1}=O_{nm}(\Lambda^{T})^{k+1},
\end{align*}
which can be written in the following form
\begin{equation*}
\left ( S_{k+1}^{I} +S_{k}^{I}R_{1}+ S_{k-1}^{I}R_{2}+\cdots + S_{2}^{I}R_{k-1}+S_{1}^{I}R_{k}+S_{k+1} \right )(\Lambda^{T})^{k+1}=O_{nm}(\Lambda^{T})^{k+1},
\end{equation*}
where
\begin{equation*}
  \begin{array}{c}
    R_{1}=Diag (O_{n},\cdots,O_{n},(S_{1})_{11},\cdots,(S_{1})_{(m-k)(m-k)}), \\
    R_{2}=Diag (O_{n},\cdots,O_{n},(S_{2})_{11},\cdots,(S_{2})_{(m-k+1)(m-k+1)}), \\
    \cdots \cdots \cdots \cdots \cdots \cdots \cdots \cdots \\
    R_{k-1}=Diag (O_{n},O_{n},(S_{k-1})_{11},\cdots,(S_{k-1})_{(m-2)(m-2)}), \\
    R_{k}=Diag (O_{n},(S_{k})_{11},\cdots,(S_{k})_{(m-1)(m-1)}), \\
  \end{array}
\end{equation*}
it shows that
\begin{equation*}
S_{k+1}^{I}=-\left (S_{k}^{I}R_{1}+ S_{k-1}^{I}R_{2}+\cdots + S_{2}^{I}R_{k-1}+S_{1}^{I}R_{k}+S_{k+1} \right ).
\end{equation*}

The proposition is proved.
\end{proof}

A matrix $S$ of the form (\ref{B4}) will be called \textbf{dressing matrix}. Thus, the set of all dressing matrices is a group under the usual product of matrices which will be denoted by $\mathfrak{G}_{-}$.

Observe that if $S\in \mathfrak{G}_{-}$ and it satisfies the block matrix linear equation
\begin{equation}\label{B6}
\frac{\partial S}{\partial t_{r}}=-H^{r}_{<}S,\,\,\,\,\,\,r=1,\cdots,m-1,
\end{equation}
where $H=S\Lambda S^{-1}$, then $H$ is solution of the equations (\ref{B3}). Conversely, let $H$ be a solution of (\ref{B3}) then any dressing block
matrix solution of (\ref{B6}) is called \textbf{Sato-Wilson} block matrix corresponding to this $H$. Since we are working in the matrix case and $H^{r}_{<}$ is strictly lower triangular block matrix for $r=1,\cdots,m-1$, such solutions of (\ref{B6}) always exist. In this case, to give a specify Sato-Wilson matrix corresponding to a Lax matrix $H$,
one only must fix an initial condition at the time of solving the equation (\ref{B6}) on the group $\mathfrak{G}_{-}$.

We have

\begin{proposition}Suppose that $H$ is solution of (\ref{B3}) for which $H(0)=\Lambda$, then there is $S\in \mathfrak{G}_{-}$
such that $H=S\Lambda S^{-1}$ where $S$ satisfies (\ref{B6}) subject to the initial condition $S(0)=I_{nm}$.
\end{proposition}
\begin{proof}Assume that $H$ satisfies (\ref{B3}) with $H(0)=\Lambda$ and let $S\in \mathfrak{G}_{-}$ be a solution of the equations (\ref{B6}) such that
$S(0)=I_{nm}$. Define $H_{S}=S^{-1}HS$ then
\begin{equation*}
\frac{\partial H_{S}}{\partial t_{k}}=\frac{\partial S^{-1}}{\partial t_{k}} HS+S^{-1}\frac{\partial H}{\partial t_{k}}S+
S^{-1}H\frac{\partial S}{\partial t_{k}}=O_{nm},
\end{equation*}
it shows that $H_{S}$ is a constant matrix. Now $H_{S}(0)=S^{-1}(0)H(0)S(0)=\Lambda$. Hence $H_{S}=\Lambda$ and $H=S\Lambda S^{-1}$.
\end{proof}

In the previous proposition the supposition $H(0)=\Lambda$ can be improved even more to include a larger class of solutions $H$ that could be written in the form $H=S\Lambda S^{-1}$
for some $S\in \mathfrak{G}_{-}$. In fact, we have
\begin{proposition}\label{callproposition}
Suppose that $H$ is solution of (\ref{B3}) such that $H(0)=S_{\iota}\Lambda S_{\iota}^{-1}$ where $S_{\iota}\in \mathfrak{G}_{-}$ is a constant dressing matrix.
Then there exists $S\in \mathfrak{G}_{-}$ such that $H=S\Lambda S^{-1}$, this decomposition is not necessarily unique.
\end{proposition}
\begin{proof}
Let $S_{a}$ be an arbitrary Sato-Wilson matrix corresponding to $H$ and define as above $H_{S_{a}}=S_{a}^{-1}HS_{a}$. Again we can see that $H_{S_{a}}$ is a constant matrix. Thus,
$H_{S_{a}}=H_{S_{a}}(0)=S_{a}^{-1}(0)H(0)S_{a}(0)$, that is, $H=S_{a}S_{a}^{-1}(0)H(0)S_{a}(0)S_{a}^{-1}$. Hence,
$$H=S_{a}S_{a}^{-1}(0)S_{\iota}\Lambda S_{\iota}^{-1}S_{a}(0)S_{a}^{-1},$$
it is enough to take $S$ as $S=S_{a}S_{a}^{-1}(0)S_{\iota}$.
\end{proof}

\subsection{Borel-Gauss factorization for block matrices and its applications}

The following lemma will be very useful below

\begin{lemma}\label{lema1}
Let $Y$ be a upper triangular block matrix of order $mn\times mn$ such that $Y_{ss}$ is nonsingular for $s=1,\ldots, m$
\begin{equation*}
Y=\left(
                           \begin{array}{ccccc}
                             Y_{11} & Y_{12} & Y_{13} & \cdots & Y_{1m} \\
                             O_{n} & \ddots & \ddots & \ddots & \vdots \\
                             \vdots & \ddots & \ddots & \ddots & Y_{(m-2)m} \\
                             \vdots & \ddots & \ddots & \ddots & Y_{(m-1)m} \\
                             O_{n} & \cdots & \cdots & O_{n} & Y_{mm} \\
                           \end{array}
                         \right),
\end{equation*}
then
\begin{equation}\label{B7}
|Y|=|Y_{11}||Y_{22}|\cdots|Y_{mm}|,
\end{equation}
therefore the matrix $Y$ is nonsingular.
\end{lemma}
\begin{proof}From Schur determinant lemma follows that (\ref{B7}) holds for $m=2$. Indeed, since $Y_{21}=O_{n}$ then $[Y_{11},Y_{21}]=O_{n}$,
where as before $[Y_{11},Y_{21}]$ indicates the Lie Bracket of $Y_{11}$ and $Y_{21}$, thus $|Y|=|Y_{11}Y_{22}-Y_{21}Y_{12}|=|Y_{11}Y_{22}|=|Y_{11}|\,\,|Y_{22}|$.
The result is also true for $m=3$. To see this we use the Schur's formula
$$|Y|=|Y_{11}|\,|(Y/Y_{11})|=|Y_{11}|\left |\begin{array}{cc}
                                              Y_{22} & Y_{23} \\
                                              O_{n} & Y_{33}
                                            \end{array}
\right |=|Y_{11}||Y_{22}||Y_{33}|.$$

We proceed now by induction. Let us suppose the result holds for $m=k$ and prove the statement of the lemma for $m=k+1$. Let $Y$ be a upper triangular block matrix of order $(k+1)n\times (k+1)n$
for which each matrix in the principal diagonal is nonsingular, then
\begin{equation*}
|Y|=|Y_{11}|\,|(Y/Y_{11})|=|Y_{11}|\,\left|
                           \begin{array}{ccccc}
                             Y_{22} & Y_{23} & Y_{24} & \cdots & Y_{2(k+1)} \\
                             O_{n} & \ddots & \ddots & \ddots & \vdots \\
                             \vdots & \ddots & \ddots & \ddots & Y_{(k-1)(k+1)} \\
                             \vdots & \ddots & \ddots & \ddots & Y_{(k)(k+1)} \\
                             O_{n} & \cdots & \cdots & O_{n} & Y_{(k+1)(k+1)} \\
                           \end{array}
                         \right|=|Y_{11}|\,|Y_{22}|\cdots |Y_{(k+1)(k+1)}|.
\end{equation*}

Hence, the lemma is fulfilled also for $m=k+1$. The proof is finished.
\end{proof}

\begin{corollary}All matrices of the form
\begin{equation}\label{B8}
S=\left(
                           \begin{array}{ccccc}
                             I_{n} & O_{n} & \cdots & \cdots & O_{n} \\
                             S_{21} & I_{n} & \ddots & \ddots & \vdots \\
                             \vdots & \ddots & \ddots & \ddots & \vdots \\
                             \vdots & \ddots & \ddots & I_{n} & O_{n} \\
                             S_{m1} & \cdots & \cdots & S_{m(m-1)} & I_{n} \\
                           \end{array}
                         \right),
                         \end{equation}
are nonsingular, even more $|S|=1$.
\end{corollary}
\begin{proof}Taking into account that $|S|=|S^{T}|$ the result is followed by previous lemma.
\end{proof}

We give the following definition
\begin{definition}We say that a nonsingular block matrix $U$ of order $mn\times mn$ admits a Borel-Gauss factorization if
\small
\begin{align} \label{B9}
&\left(
                           \begin{array}{ccccc}
                             U_{11} & U_{12} & \cdots & \cdots & U_{1m} \\
                             U_{21} & \ddots & \ddots & \ddots & \vdots \\
                             \vdots & \ddots & \ddots & \ddots & \vdots \\
                             \vdots & \ddots & \ddots & \ddots & U_{(m-1)m} \\
                             U_{m1} & \cdots & \cdots & U_{m(m-1)} & U_{mm} \\
                           \end{array}
                         \right) \\ \nonumber
&=\left(
                           \begin{array}{ccccc}
                             I_{n} & O_{n} & \cdots & \cdots & O_{n} \\
                             S_{21} & I_{n} & \ddots & \ddots & \vdots \\
                             \vdots & \ddots & \ddots & \ddots & \vdots \\
                             \vdots & \ddots & \ddots & I_{n} & O_{n} \\
                             S_{m1} & \cdots & \cdots & S_{m(m-1)} & I_{n} \\
                           \end{array}
                         \right)\left(
                           \begin{array}{ccccc}
                             Y_{11} & Y_{12} & Y_{13} & \cdots & Y_{1m} \\
                             O_{n} & \ddots & \ddots & \ddots & \vdots \\
                             \vdots & \ddots & \ddots & \ddots & Y_{(m-2)m} \\
                             \vdots & \ddots & \ddots & \ddots & Y_{(m-1)m} \\
                             O_{n} & \cdots & \cdots & O_{n} & Y_{mm} \\
                           \end{array}
                         \right),
\end{align}

where $|Y_{kk}|\neq 0$, for $k=1,2,\cdots, m$.
\end{definition}
\normalsize

Denote
\begin{equation}
\Delta_{k}(U)=\left|
              \begin{array}{ccc}
                U_{11} & \cdots & U_{1k} \\
                \vdots & \ddots & \vdots \\
                U_{k1} & \cdots & U_{kk} \\
              \end{array}
            \right|,
\end{equation}
for $k=1,2,\cdots,m$. The square sub matrices giving place to the determinants $\Delta_{k}(U)$ are called the main minors of the matrix $U$ and they are denoted by $M_{k}(U)$. Thus, $\Delta_{k}(U)=|M_{k}(U)|$. We have
\begin{theorem}Let us suppose that $\Delta_{k}(U)\neq 0$ for $k=1,2,\ldots,m$ where $m\geq 2$, then $U$ admits a factorization of Borel-Gauss type.
\end{theorem}
\begin{proof}As before we do the proof by induction. Suppose that $m=2$, in this case we must prove that there exists $S$ lower triangular block matrix of order $2n\times 2n$ and $Y$ upper triangular block matrix of the same order such that
\begin{equation*}
U=\left(
    \begin{array}{cc}
      U_{11} & U_{12} \\
      U_{21} & U_{22} \\
    \end{array}
  \right)=\left(
            \begin{array}{cc}
              I_{n} & O_{n} \\
              S_{21} & I_{n} \\
            \end{array}
          \right)\left(
                   \begin{array}{cc}
                     Y_{11} & Y_{12} \\
                     O_{n} & Y_{22} \\
                   \end{array}
                 \right),
\end{equation*}
where $U$ is a matrix for which $\Delta_{1}(U)=|M_{1}(U)|=|U_{11}|\neq 0$ and $\Delta_{2}(U)=|M_{2}(U)|=|U|\neq 0$. The following calculation is well known: $Y_{11}=U_{11}$ and $Y_{12}=U_{12}$. Moreover,
$S_{21}=U_{21}U_{11}^{-1}$ and finally $Y_{22}=U_{22}-U_{21}U_{11}^{-1}U_{12}=M_{2}(U)/M_{1}(U)$. It shows that $|Y_{22}|\neq 0$. In fact, from Schur's formula
\begin{equation*}
|Y_{22}|=|M_{2}(U)/M_{1}(U)|=\frac{|M_{2}(U)|}{|M_{1}(U)|}=\frac{\Delta_{2}}{\Delta_{1}}=\frac{|U|}{|U_{11}|}\neq 0.
\end{equation*}

We would like to calculate the Borel-Gauss factorization for a matrix $U$ of order $3n\times 3n$ such that $\Delta_{3}$, $\Delta_{2}$ and $\Delta_{1}$ are different from zero. But before this, observe that one can calculate the entries of the matrices $S$ and $Y$ in a recurring way (from the inside out) taking into account
\begin{equation}\label{B10}
M_{k}(U)=M_{k}(S)M_{k}(Y),
\end{equation}
for $k=1,2,\cdots,m$. Next, we do the computation for $m=3$, that is, we must have
\begin{equation*}
U=\left(
    \begin{array}{ccc}
      U_{11} & U_{12} & U_{13} \\
      U_{21} & U_{22} & U_{23} \\
      U_{31} & U_{32} & U_{33}
    \end{array}
  \right)=\left(
            \begin{array}{ccc}
              I_{n} & O_{n} & O_{n} \\
              S_{21} & I_{n} & O_{n} \\
              S_{31} & S_{32} & I_{n}
            \end{array}
          \right)\left(
                   \begin{array}{ccc}
                     Y_{11} & Y_{12} & Y_{13} \\
                     O_{n} & Y_{22} & Y_{23} \\
                     O_{n} & O_{n} & Y_{33}
                   \end{array}
                 \right),
\end{equation*}
under the supposition that $M_{1}(U)$, $M_{2}(U)$ and $M_{3}(U)$ are nonsingular square matrices. Since $M_{2}(U)=M_{2}(S)M_{2}(Y)$ then we already know how to calculate the entries of $M_{2}(S)$ and $M_{2}(Y)$ from the entries of $M_{2}(U)$. Thus,
\begin{equation*}
Y_{11}=U_{11},\,\,Y_{12}=U_{12},\,\,Y_{22}=U_{22}-U_{21}U_{11}^{-1}U_{12},\,\,S_{21}=U_{21}U_{11}^{-1}
\end{equation*}
and so $|Y_{11}|\neq 0$ and $|Y_{22}|\neq 0$. On other hand, recalling that $|Y_{22}|=|U_{22}-U_{21}U_{11}^{-1}U_{12}|\neq 0$
\begin{align}
&Y_{13}=U_{13},\,\,Y_{23}=U_{23}-U_{21}U_{11}^{-1}U_{13},\,\,S_{31}=U_{31}U_{11}^{-1},\,\,S_{32}=(U_{32}-U_{31}U_{11}^{-1}U_{12})
(U_{22}-U_{21}U_{11}^{-1}U_{12})^{-1}, \nonumber \\
&Y_{33}=(U_{33}-U_{31}U_{11}^{-1}U_{13}) -(U_{32}-U_{31}U_{11}^{-1}U_{12})(U_{22}-U_{21}U_{11}^{-1}U_{12})^{-1}(U_{23}-U_{21}U_{11}^{-1}U_{13}).  \nonumber
\end{align}

We claim that
\begin{equation*}
Y_{33}=M_{3}(U)/M_{2}(U)=U_{33}-\left(
                                  \begin{array}{cc}
                                    U_{31} & U_{32} \\
                                  \end{array}
                                \right)\left(
                                         \begin{array}{cc}
                                           U_{11} & U_{12} \\
                                           U_{21} & U_{22} \\
                                         \end{array}
                                       \right)^{-1}\left(
                                                     \begin{array}{c}
                                                       U_{13} \\
                                                       U_{23} \\
                                                     \end{array}
                                                   \right),
\end{equation*}
indeed, a simple calculation shows that
\begin{equation*}
\left(
                                         \begin{array}{cc}
                                           U_{11} & U_{12} \\
                                           U_{21} & U_{22} \\
                                         \end{array}
                                       \right)^{-1}=\left(
                                                      \begin{array}{cc}
                                                        U_{11}^{-1}+U_{11}^{-1}U_{12}(U_{22}-U_{21}U_{11}^{-1}U_{12})^{-1}U_{21}U_{11}^{-1} & -U_{11}^{-1}U_{12}(U_{22}-U_{21}U_{11}^{-1}U_{12})^{-1} \\
                                                        -(U_{22}-U_{21}U_{11}^{-1}U_{12})^{-1}U_{21}U_{11}^{-1} &  (U_{22}-U_{21}U_{11}^{-1}U_{12})^{-1} \\
                                                      \end{array}
                                                    \right),
\end{equation*}
thus
\small
\begin{align*}
&\left(
                                         \begin{array}{cc}
                                           U_{11} & U_{12} \\
                                           U_{21} & U_{22} \\
                                         \end{array}
                                       \right)^{-1}\left(
                                                     \begin{array}{c}
                                                       U_{13} \\
                                                       U_{23} \\
                                                     \end{array}
                                                   \right) \\ &=\left(
                                                      \begin{array}{c}
                                                        U_{11}^{-1}U_{13}+U_{11}^{-1}U_{12}(U_{22}-U_{21}U_{11}^{-1}U_{12})^{-1}U_{21}U_{11}^{-1}U_{13} -U_{11}^{-1}U_{12}(U_{22}-U_{21}U_{11}^{-1}U_{12})^{-1}U_{23} \\
                                                        -(U_{22}-U_{21}U_{11}^{-1}U_{12})^{-1}U_{21}U_{11}^{-1}U_{13}+
                                                        (U_{22}-U_{21}U_{11}^{-1}U_{12})^{-1}U_{23} \\
                                                      \end{array}
                                                    \right), \\
\end{align*}
\normalsize
from here it is easy to prove the affirmation. It implies
\begin{equation*}
|Y_{33}|=|M_{3}(U)/M_{2}(U)|=\frac{|M_{3}(U)|}{|M_{2}(U)|}=\frac{\Delta_{3}}{\Delta_{2}}\neq 0.
\end{equation*}

Suppose the theorem holds for $k=m$ and let us show that this is also true for $k=m+1$. Let $U$ be a block matrix of order $(m+1)n\times (m+1)n$ such that $\Delta_{1}\neq 0,\Delta_{2}\neq 0,\cdots,\Delta_{m}\neq 0,\Delta_{m+1}\neq 0$, then by the induction hypothesis $M_{m}(U)$ admits a Borel-Gauss factorization, that is, $M_{m}(U)=S_{m}(U)Y_{m}(U)$ and $Y_{m}(U)$ having its main diagonal composed of non-singular matrices. Denote
\begin{equation*}
S_{m}(U)=\left(
                           \begin{array}{ccccc}
                             I_{n} & O_{n} & \cdots & \cdots & O_{n} \\
                             S_{21} & I_{n} & \ddots & \ddots & \vdots \\
                             \vdots & \ddots & \ddots & \ddots & \vdots \\
                             \vdots & \ddots & \ddots & I_{n} & O_{n} \\
                             S_{m1} & \cdots & \cdots & S_{m(m-1)} & I_{n} \\
                           \end{array}
                         \right),\,\,\,Y_{m}(U)=\left(
                           \begin{array}{ccccc}
                             Y_{11} & Y_{12} & Y_{13} & \cdots & Y_{1m} \\
                             O_{n} & \ddots & \ddots & \ddots & \vdots \\
                             \vdots & \ddots & \ddots & \ddots & Y_{(m-2)m} \\
                             \vdots & \ddots & \ddots & \ddots & Y_{(m-1)m} \\
                             O_{n} & \cdots & \cdots & O_{n} & Y_{mm} \\
                           \end{array}
                         \right).
\end{equation*}

Then we can find the matrices
$$S_{(m+1)1}, S_{(m+1)2},\cdots,S_{(m+1)(m-1)},S_{(m+1)m}$$
and
$$Y_{1(m+1)},Y_{2(m+1)},\cdots, Y_{m(m+1)},Y_{(m+1)(m+1)}$$
such that
\begin{equation*}
U=\left(
                           \begin{array}{cccccc}
                             I_{n} & O_{n} & \cdots & \cdots & O_{n} & O_{n} \\
                             S_{21} & I_{n} & \ddots & \ddots & \ddots & \vdots\\
                             \vdots & \ddots & \ddots & \ddots & \ddots & \vdots \\
                             \vdots & \ddots & \ddots & I_{n} & O_{n} & \vdots \\
                             S_{m1} & \cdots & \cdots & S_{m(m-1)} & I_{n} & O_{n} \\
                             S_{(m+1)1} & \cdots & \cdots & S_{(m+1)(m-1)} & S_{(m+1)m} & I_{n}
                           \end{array}
                         \right)\left(
                           \begin{array}{cccccc}
                             Y_{11} & Y_{12} & Y_{13} & \cdots & Y_{1m} & Y_{1(m+1)} \\
                             O_{n} & \ddots & \ddots & \ddots & \ddots & \vdots \\
                             \vdots & \ddots & \ddots & \ddots & \ddots & \vdots \\
                             \vdots & \ddots & \ddots & \ddots & \ddots & Y_{(m-1)(m+1)} \\
                             \vdots & \cdots & \cdots & \cdots & Y_{mm} & Y_{m(m+1)} \\
                             O_{n} & \cdots & \cdots & \cdots & O_{n} & Y_{(m+1)(m+1)}
                           \end{array}
                         \right).
\end{equation*}

In fact, we have
\begin{equation}\label{B11}
\left(
  \begin{array}{c}
    Y_{1(m+1)} \\
     Y_{2(m+1)} \\
    \vdots \\
    Y_{(m-1)(m+1)} \\
    Y_{m(m+1)} \\
  \end{array}
\right)=\left(
                           \begin{array}{ccccc}
                             I_{n} & O_{n} & \cdots & \cdots & O_{n} \\
                             S_{21} & I_{n} & \ddots & \ddots & \vdots \\
                             \vdots & \ddots & \ddots & \ddots & \vdots \\
                             \vdots & \ddots & \ddots & I_{n} & O_{n} \\
                             S_{m1} & \cdots & \cdots & S_{m(m-1)} & I_{n} \\
                           \end{array}
                         \right)^{-1}\left(
  \begin{array}{c}
    U_{1(m+1)} \\
     U_{2(m+1)} \\
    \vdots \\
    U_{(m-1)(m+1)} \\
    U_{m(m+1)} \\
  \end{array}
\right),
\end{equation}
and
\begin{align}\label{B12}
&\left(
  \begin{array}{ccc}
    S_{(m+1)1}\,\,S_{(m+1)2}\,\, \cdots\,\,  S_{(m+1)m} \\
  \end{array}
\right) \nonumber \\ &=\left(
  \begin{array}{ccccc}
    U_{(m+1)1}\,\, U_{(m+1)2}\,\, \cdots\,\, U_{(m+1)m} \\
  \end{array}
\right)\left(
                           \begin{array}{ccccc}
                             Y_{11} & Y_{12} & Y_{13} & \cdots & Y_{1m} \\
                             O_{n} & \ddots & \ddots & \ddots & \vdots \\
                             \vdots & \ddots & \ddots & \ddots & Y_{(m-2)m} \\
                             \vdots & \ddots & \ddots & \ddots & Y_{(m-1)m} \\
                             O_{n} & \cdots & \cdots & O_{n} & Y_{mm} \\
                           \end{array}
                         \right)^{-1}.
\end{align}

Finally, notice that necessarily
\small
\begin{equation*}
Y_{(m+1)(m+1)}=U_{(m+1)(m+1)}-\sum_{k=1}^{m}S_{(m+1)k}Y_{k(m+1)},
\end{equation*}
\normalsize
and therefore
\begin{equation*}
Y_{(m+1)(m+1)}=M_{m+1}(U)/M_{m}(U),
\end{equation*}
here, we have used (\ref{B11}) and (\ref{B12}). It shows that
$$|Y_{(m+1)(m+1)}|=\frac{|M_{m+1}(U)|}{|M_{m}(U)|}=\frac{\Delta_{(m+1)}}{\Delta_{m}}\neq 0.$$

We conclude the proof of the theorem.
\end{proof}

\begin{remark}Observe that if $U$ is a block matrix of order $mn\times mn$ which admits a Borel-Gauss factorization then
necessarily $\Delta_{k}(U)\neq 0$ for $k=1,2,\cdots,m$.
\end{remark}

\begin{remark}Let $U$ be a block matrix of order $mn\times mn$ admitting a Borel-Gauss factorization, then $Y_{11}=U_{11}$ and
\begin{equation}\label{B13}
Y_{kk}=M_{k}(U)/M_{k-1}(U),
\end{equation}
for $k=2,\cdots,m$.
\end{remark}

We must mention that the theme of the Borel-Gauss factorization for semi-infinite moments block matrices was investigated in the papers
\cite{alvarez} and \cite{ariznabarrete}. Let us return to our study of the hierarchy (\ref{B3}). The following result is fundamental in
the study of the block matrix finite discrete KP hierarchy.

\begin{lemma}\label{otro llamado}
Suppose that $U(t_{1},\cdots,t_{m-1})$ admits a Borel-Gauss factorization $U=S_{U}^{-1}Y_{U}$ such that $S_{U}$ and $Y_{U}$ satisfy
the linear equations
\begin{equation}\label{B14}
\frac{\partial S_{U}}{\partial t_{k}}=-H^{k}_{<}S_{U},\,\,\,\,\,\,\frac{\partial Y_{U}}{\partial t_{k}}=H^{k}_{\geq}Y_{U},\,\,\,\,\,\,\,\,k=1,2,\cdots,m-1,
\end{equation}
where $H=S_{U}\Lambda S_{U}^{-1}$, with initial conditions $S_{U}(0)=I_{nm}$ and $Y_{U}(0)=I_{nm}$ (observe that $S_{U}\in\mathfrak{G}_{-}$).
Then, $U=e^{\sum_{k=1}^{m-1}\Lambda^{k}t_{k}}$.
\end{lemma}

\begin{proof}First of all, we have $U(0)=I_{nm}$. On other hand
\begin{align*}
\frac{\partial U}{\partial t_{k}}&=\frac{\partial S_{U}^{-1}}{\partial t_{k}}Y_{U}+S_{U}^{-1}\frac{\partial Y_{U}}{\partial t_{k}}
=-S_{U}^{-1}\frac{\partial S_{U}}{\partial t_{k}}S_{U}^{-1}Y_{U}+S_{U}^{-1}H^{k}_{\geq}Y_{U}=S_{U}^{-1}H^{k}_{<}Y_{U}+S_{U}^{-1}H^{k}_{\geq}Y_{U} \\
&=S_{U}^{-1}H^{k}Y_{U}=S_{U}^{-1}H^{k}S_{U}U=\Lambda^{k}U,
\end{align*}
it implies that $U=e^{\sum_{k=1}^{m-1}\Lambda^{k}t_{k}}$.
\end{proof} \\

We have that $U=e^{\sum_{k=1}^{m-1}\Lambda^{k}t_{k}}$ is a block upper triangular matrix, denote this matrix function for $Y_{E}(t_{1},\cdots,t_{m-1})$.

From now on, the set of all matrices of order $mn \times mn$ admitting a Borel-Gauss factorization will be denoted for $\mathfrak{BG}(mn)$. In general, if
$U\in \mathfrak{BG}(mn)$ such that $U=S_{U}^{-1}Y_{U}$, for which $S_{U}$ and $Y_{U}$ satisfy the equations (\ref{B14}) where $H=S_{U}\Lambda S_{U}^{-1}$, then
\begin{equation*}
U=e^{\sum_{k=1}^{m-1}\Lambda^{k}t_{k}}U(0)=e^{\sum_{k=1}^{m-1}\Lambda^{k}t_{k}}S_{U}^{-1}(0)Y_{U}(0)=Y_{E}(t_{1},\cdots,t_{m-1})S_{U}^{-1}(0)Y_{U}(0).
\end{equation*}

Observe that
\begin{equation*}
\mathfrak{BG}(mn)=\left\{ S_{1}^{-1}J (S_{2}^{-1})^{T}| S_{1}^{-1},S_{2}^{-1}\in\mathfrak{G}_{-}, J\in \mathfrak{D}(mn)\right\},
\end{equation*}
where $\mathfrak{D}(mn)$ is the space of all block diagonal matrices $J$ such that $\left | J_{ii}\right|\neq 0$ for $i=1,\ldots, m$.

\begin{theorem}\label{teoimp1}
Let us suppose that $U\in \mathfrak{BG}(mn)$ and $U=S_{U}^{-1}Y_{U}$ its Gauss-Borel factorization, such that
\begin{equation}\label{B15}
\frac{\partial S_{U}}{\partial t_{k}}=-H^{k}_{<}S_{U},
\end{equation}
for $k=1,2,\cdots,m-1$, where $H=S_{U}\Lambda S_{U}^{-1}$. Then $Y_{U}$ satisfies the equations
\begin{equation}\label{B16}
\frac{\partial Y_{U}}{\partial t_{k}}=H^{k}_{\geq}Y_{U},\,\,\,\,\,\,\,\,k=1,2,\cdots,m-1,
\end{equation}
if and only if, $U$ satisfies
\begin{equation}\label{B17}
\frac{\partial U}{\partial t_{k}}=\Lambda^{k}U,
\end{equation}
for $k=1,2,\cdots,m-1$.
\end{theorem}
\begin{proof} One of the statements follows from the previous lemma. Now suppose that
\begin{equation*}
\frac{\partial S_{U}}{\partial t_{k}}=-H^{k}_{<}S_{U},\,\,\,\,\,\,\,\,\,\,\frac{\partial U}{\partial t_{k}}=\Lambda^{k}U,
\end{equation*}
where $k=1,2,\cdots,m-1$. Then
\begin{equation*}
\Lambda^{k}U=\frac{\partial U}{\partial t_{k}}=\frac{\partial (S_{U}^{-1}Y_{U})}{\partial t_{k}}=\frac{\partial S_{U}^{-1}}{\partial t_{k}}Y_{U}+S_{U}^{-1}\frac{\partial Y_{U}}{\partial t_{k}}
=-S_{U}^{-1}\frac{\partial S_{U}}{\partial t_{k}}S_{U}^{-1}Y_{U}+S_{U}^{-1}\frac{\partial Y_{U}}{\partial t_{k}},
\end{equation*}
so
\begin{equation*}
\frac{\partial Y_{U}}{\partial t_{k}}=H^{k}Y_{U}+\frac{\partial S_{U}}{\partial t_{k}}S_{U}^{-1}Y_{U}=H^{k}Y_{U}-H^{k}_{<}Y_{U}=H^{k}_{\geq}Y_{U}.
\end{equation*}
hence, we obtain (\ref{B16}).
\end{proof}

Really, the three equations (\ref{B15})-(\ref{B17}) are such that if at least two of them are true, then one can check that the third equation holds.

\section{Properties of parametric linear systems related to the solutions of the block matrix finite discrete KP hierarchy}
\label{Properties}

\subsection{Linear system associated with $H(0)=S_{\iota}\Lambda S^{-1}_{\iota}$ where $H(t)$ is a solution of (\ref{B3})}
\label{subseccion3}

We will start by studying linear systems associated with the initial conditions $H(0)=S_{\iota}\Lambda S^{-1}_{\iota}$ of solutions $H(t)$ of the hierarchy (\ref{B3}). First, we suppose that $S_{\iota}=I_{nm}$, that is, $H(0)=\Lambda$ is the simplest solution of (\ref{B3}).

Let $\Lambda$ be the $mn\times mn$ block matrix shift given for (\ref{B1}). Following \cite{CZ}, let us introduce a linear state system on the space $\mathcal{X^{\tau}}$ (denominated state space)
of matrices $x(\tau)=(x_{1}(\tau),\cdots,x_{m}(\tau))^{T}$ where each $x_{k}(\tau)$ is of order $n\times n$ for $k=1,\cdots,m$, of the form
\begin{align}\label{SDL0}
\frac{d\,x(\tau)}{d \tau} &= \Lambda x(\tau)+ Bv(\tau), \nonumber \\
y(\tau) &= Cx(\tau), \qquad \qquad \qquad \tau \geq 0, \quad x(0)=x_{0}.
\end{align}
where $B$ is a control column vector of order $mn\times n$ which will be specified below. On other hand, $v(\tau)$ belongs to the space $\mathcal{V}$ (denominated input space) of $n\times n$ matrices.
Specifically

\begin{equation}\label{B}
             B=\left(
                \begin{array}{c}
                   M_{-1} \\
                   M_{-2} \\
                   \vdots \\
                   M_{-(m-1)} \\
                   M_{-m} \\
                \end{array}
              \right),
\end{equation}
where the $M_{-k}$, for $k=1,\cdots, m$ are arbitrary constant matrices of order $n\times n$. Moreover $C=(I_{n},O_{n},\dots ,O_{n})$ is the $n\times mn$
observation row vector.

The transfer function of the system (\ref{SDL0}) is the Laurent polynomial
\begin{equation}\label{lp}
F_{0}(z)= C(zI_{nm}-\Lambda)^{-1} B= \frac{M_{-1}z^{m-1}+\cdots+M_{-m}}{z^{m}},
\end{equation}
where $z\neq 0 \in\mathbb{C}$.

\begin{proposition} If $|M_{-m}|\neq 0$ the linear dynamical system (\ref{SDL0}) is controllable. In any case, it is observable.
\end{proposition}
\begin{proof}We must prove that
\begin{equation}\label{nanzov1}
rank(B\,\,\,\Lambda B\,\,\,...\,\,\,\,\Lambda ^{m-1}B)=mn.
\end{equation}

Denote $\Gamma=(B\,\,\,\Lambda B\,\,\,...\,\,\,\,\Lambda ^{m-1}B)$, then (\ref{nanzov1}) is equivalent to the condition
\begin{equation*}
|\Gamma|=\left |
\begin{array}{ccccc}
M_{-1} & M_{-2} & \cdots & M_{-(m-1)} & M_{-m} \\
M_{-2} & M_{-3} & \cdots & M_{-m} & O_{n} \\
\vdots & \vdots & \ddots & O_{n} & \vdots \\
M_{-(m-1)} & M_{-m} & O_{n} & \cdots & \vdots \\
M_{-m} & O_{n} & \cdots & \cdots & O_{n} \\
\end{array}
 \right |\neq 0,
\end{equation*}
what is equivalent in turn to the next
\begin{equation*}
|\Gamma_{\ast}|=\left |
\begin{array}{ccccc}
M_{-m} & O_{n} & \cdots & \cdots & O_{n} \\
M_{-(m-1)} & M_{-m} & O_{n} & \cdots & \vdots \\
\vdots & \vdots & \ddots & O_{n} & \vdots \\
M_{-2} & M_{-3} & \cdots & M_{-m} & O_{n} \\
M_{-1} & M_{-2} & \cdots & M_{-(m-1)} & M_{-m} \\
\end{array}
 \right |\neq 0,
\end{equation*}
this is because $|\Gamma|=-|\Gamma_{\ast}|=-|(\Gamma_{\ast})^{T}|$. Hence from the assumption $|M_{-m}|\neq 0$ and the lemma \ref{lema1}
we conclude $|\Gamma_{\ast}|=(|M_{-m}^{T}|)^{m}=(|M_{-m}|)^{m}\neq 0$. Thus, the system (\ref{SDL0}) is controllable.

To prove that the system (\ref{SDL0}) is observable it is sufficient to show that
\begin{equation}\label{nanzov2}
rank(D \,\,\, \Lambda^{T}D\,\,...\,\, (\Lambda ^{T})^{m-1}D)=mn,
\end{equation}
where $D=C^{T}$. We claim that (\ref{nanzov2}) holds. Indeed, $(D \,\,\, \Lambda^{T}D\,\,...\,\, (\Lambda ^{T})^{m-1}D)$ is the identity matrix
of order $mn\times mn$. So, $|D \,\,\, \Lambda^{T}D\,\,...\,\, (\Lambda ^{T})^{m-1}D|=1$. It implies (\ref{nanzov2}) and that (\ref{SDL0})
is observable.
\end{proof}

\begin{remark}For all $S_{\iota}\in \mathfrak{G}_{-}$, we have $CS_{\iota}=CS^{-1}_{\iota}=C$.
\end{remark}

Taking into account the proposition \ref{callproposition} and the previous remark we can consider a more general linear system
\begin{align}\label{SDLG0}
\frac{d\,x(\tau)}{d \tau} &= S_{\iota}\Lambda S^{-1}_{\iota} x(\tau)+ Bv(\tau), \nonumber \\
y(\tau) &= Cx(\tau), \qquad \qquad \qquad \tau \geq 0, \quad x(0)=x_{0},
\end{align}
where $B$ and $C$ have the same meaning as in the system (\ref{SDL0}).

First of all observe that the transfer function of (\ref{SDLG0}) is of the form (\ref{lp}). Indeed
\begin{align*}
F_{\iota}(z)&=C(zI_{nm}-S_{\iota}\Lambda S^{-1}_{\iota})^{-1} B=C(zS_{\iota}S^{-1}_{\iota}-S_{\iota}\Lambda S^{-1}_{\iota})^{-1} B=
C S_{\iota}(zI_{nm}-\Lambda)^{-1} S^{-1}_{\iota}B \\
&=C(z\mathcal{I}-\Lambda)^{-1} S^{-1}_{\iota}B=\frac{M^{\iota}_{-1}z^{m-1}+\cdots+M^{\iota}_{-m}}{z^{m}},
\end{align*}
where $M^{\iota}_{-k}$ for $k=1,\cdots,m$ are certain matrices of order $n\times n$. We have the following
\begin{proposition}The system (\ref{SDLG0}) is controllable if
\begin{equation*}
|S^{-\iota}_{m1}M_{-1}+S^{-\iota}_{m2}M_{-2}+\cdots+S^{-\iota}_{m(m-1)}M_{-(m-1)}+M_{-m}|\neq 0,
\end{equation*}
where
\begin{equation*}
B=\left(
                \begin{array}{c}
                   M_{-1} \\
                   M_{-2} \\
                   \vdots \\
                   M_{-(m-1)} \\
                   M_{-m} \\
                \end{array}
              \right),\,\,\,\,\,\,\,S^{-1}_{\iota}=\left(
                           \begin{array}{ccccc}
                             I_{n} & O_{n} & \cdots & \cdots & O_{n} \\
                             S^{-\iota}_{21} & I_{n} & \ddots & \ddots & \vdots \\
                             \vdots & \ddots & \ddots & \ddots & \vdots \\
                             \vdots & \ddots & \ddots & I_{n} & O_{n} \\
                             S^{-\iota}_{m1} & \cdots & \cdots & S^{-\iota}_{m(m-1)} & I_{n} \\
                           \end{array}
                         \right)
\end{equation*}
and observable in any case.
\end{proposition}
\begin{proof}The system (\ref{SDLG0}) is clearly observable, it remains to see that it is controllable. As usual, we need prove that
\begin{equation}\label{nanzovg1}
rank(B\,\,\,S_{\iota}\Lambda S^{-1}_{\iota} B\,\,\,...\,\,\,\,S_{\iota}\Lambda ^{m-1}S^{-1}_{\iota} B)=
rank(S_{\iota}S^{-1}_{\iota}B\,\,\,S_{\iota}\Lambda S^{-1}_{\iota} B\,\,\,...\,\,\,\,S_{\iota}\Lambda ^{m-1}S^{-1}_{\iota} B)=mn.
\end{equation}

Since $|S_{\iota}|=1$, the condition (\ref{nanzovg1}) is equivalent to the following condition
\begin{equation*}
rank(S^{-1}_{\iota}B\,\,\,\Lambda S^{-1}_{\iota} B\,\,\,...\,\,\,\,\Lambda ^{m-1}S^{-1}_{\iota} B)
=rank(B_{\iota}\,\,\,\Lambda B_{\iota}\,\,\,...\,\,\,\,\Lambda ^{m-1}B_{\iota})=mn,
\end{equation*}
where $B_{\iota}=S^{-1}_{\iota}B$. Now, taking into account that the last component of $B_{\iota}$ is precisely $S^{-\iota}_{m1}M_{-1}+S^{-\iota}_{m2}M_{-2}+\cdots+S^{-\iota}_{m(m-1)}M_{-(m-1)}+M_{-m}$, the proof of this last
condition is done in a similar form to the proof of (\ref{nanzov1}) of the previous proposition.
\end{proof}

\subsection{General linear systems related to the solutions $H(t)$ of the hierarchy (\ref{B3})}
\label{subsection4}

Let us consider a parametric linear state system on the space $\mathcal{X^{\tau}}$, with parameters $t=(t_1,...,t_m)$, given as
\begin{align}\label{SDL}
\frac{d\,x(\tau,t)}{d \tau} &= H(t) x(\tau,t)+ B(t)v(\tau,t), \nonumber  \\
y(\tau,t) &= C(t)x(\tau,t), \qquad \qquad \qquad \tau \geq 0, \quad x(0,t)=x_{0}(t).
\end{align}
where $B(t)$ is the control column vector considered as a transformation from $\mathcal{V}$ to $\mathcal{X^{\tau}}$ and $C(t)$
is the observation row vector. Both vectors and the block matrix $H(t)$ are defined according to the following proposition

\begin{proposition}\label{alerta}
Suppose that $U(t)\in \mathfrak{BG}(mn)$, $U=S_{U}^{-1}(t)Y_{U}(t)$ being the equations (\ref{B14}) hold
and $Y_{U}(0)=I_{nm}$. Define $B(t)$ and $C(t)$ of the following form
$B(t)=Y_{U}(t)B$ and $C(t)=D(t)^{T}$ where $D(t)=(S_{U}^{-1}(t))^{T}D$. Here, $B$ and $D$ are defined as in the previous subsection.
Then, $B(t)$ and $D(t)$ satisfy the linear equations
\begin{equation}\label{nanzov3}
\frac{\partial B(t)}{\partial t_{k}}=H_{\geq}^{k}(t)B(t)\ ,\
\frac{\partial D(t)}{\partial t_{k}}=(H_{<}^{k}(t))^{T}D(t) ,\qquad k=1,...,m-1,
\end{equation}
with initial conditions $B(0)=B$ and $D(0)=D$. We recall that as usual $H(t)=S_{U}\Lambda S_{U}^{-1}$.
\end{proposition}
\begin{proof}For $k=1,\cdots,m-1$, we have
\begin{equation*}
\frac{\partial B(t)}{\partial t_{k}}=\frac{\partial Y_{U}(t)}{\partial t_{k}}B=H_{\geq}^{k}(t)Y_{U}(t)B=H_{\geq}^{k}(t)B(t),
\end{equation*}
and clearly $B(0)=Y_{U}(0)B=B$. On other hand,
\begin{equation*}
\frac{\partial D(t)}{\partial t_{k}}=\frac{\partial (S_{U}^{-1}(t))^{T}}{\partial t_{k}}D=\frac{\partial (S_{U}^{T}(t))^{-1}}{\partial t_{k}}D,
\end{equation*}
and taking into account that
\begin{equation*}
\frac{\partial S_{U}^{T}(t)}{\partial t_{k}}=-S_{U}^{T}(t)(H_{<}^{k}(t))^{T},
\end{equation*}
then, combining the two previous equations, we arrive to the following equality
\begin{equation*}
\frac{\partial D(t)}{\partial t_{k}}=-(S_{U}^{T}(t))^{-1}\frac{\partial S_{U}^{T}(t)}{\partial t_{k}}(S_{U}^{T}(t))^{-1}D
=(H_{<}^{k}(t))^{T}(S_{U}^{T}(t))^{-1}D=(H_{<}^{k}(t))^{T}D(t).
\end{equation*}

Observe that $(S_{U}^{-1}(t))^{T}$ is a block upper triangular matrix whose diagonal is formed with the identity matrix of order $n$. Hence, we can verify $D(0)=(S_{U}^{-1}(0))^{T}D=D$. It is interesting to observe that independently of the choice of $S(0)^{-1}$ as initial condition in the factorization of $U$, the flow for $D(t)$ always begins in $D$.
\end{proof}

The transfer function of the system (\ref{SDL}) will be the following matrix-valued function as a function of $t$
\begin{equation} \label{ftSDL}
F(z,t) = C(t)(zI_{nm}-H(t))^{-1} B(t).
\end{equation}

\begin{remark}The transfer function (\ref{ftSDL}) has the form
\begin{equation}\label{nanzov4}
F(z,t)=\frac{Q_{-1}(t)z^{m-1}+\cdots+Q_{-m}(t)}{z^{m}},
\end{equation}
where the $Q_{-k}(t)$, for $k=1, \dots,m-1, m$ are certain $n\times n$ matrices, that is, it is a Laurent polynomial. In fact,
\begin{align*}
F(z,t) &= C(t)(zI_{nm}-H(t))^{-1} B(t)=(D(t))^{T}\left ( zS_{U}(t)S^{-1}_{U}(t)-S_{U}(t)\Lambda S^{-1}_{U}(t) \right )^{-1}Y_{U}(t)B \\
&=C S^{-1}_{U}(t)\left ( zS_{U}(t)S^{-1}_{U}(t)-S_{U}(t)\Lambda S^{-1}_{U}(t) \right )^{-1}Y_{U}(t)B=C(zI_{nm}-\Lambda)^{-1}U(t)B.
\end{align*}
\end{remark}

\begin{theorem}
In order to recover a Laurent polynomial $L(t,z)$ of the form
\begin{equation*}
L(t,z)=\frac{P_{-1}(t)z^{m-1}+\cdots+P_{-m}(t)}{z^{m}},
\end{equation*}
as the transfer function of a linear system (\ref{SDL}) with $C(t)$, $B(t)$ and $H(t)$  defined as in the proposition \ref{alerta}, is necessary and sufficient that there exists a block matrix-valued function $U(t)\in \mathfrak{BG}(mn)$ for all $t\geq 0$ such that $$(P^{T}_{-1}(t),P^{T}_{-2}(t),\cdots, P^{T}_{-(m-1)}(t),P^{T}_{-m}(t))^{T}=U(t)B.$$
In this case, $$B=(U(0))^{-1}(P^{T}_{-1}(0),P^{T}_{-2}(0),\cdots, P^{T}_{-(m-1)}(0),P^{T}_{-m}(0))^{T}.$$
\end{theorem}
\begin{proof}Let us assume that
$$(P^{T}_{-1}(t),P^{T}_{-2}(t),\cdots, P^{T}_{-(m-1)}(t),P^{T}_{-m}(t))^{T}=U(t)B$$ for some block matrix-valued function
$U(t)\in \mathfrak{BG}(mn)$ for all $t\geq 0$. It implies $$B=(U(0))^{-1}(P^{T}_{-1}(0),P^{T}_{-2}(0),\cdots, P^{T}_{-(m-1)}(0),P^{T}_{-m}(0))^{T}.$$
Since $U(t)=S^{-1}_{U}(t)Y_{U}(t)$, then from the computation performed in the previous remark follows
\begin{align*}
L(t,z)&=C(zI_{nm}-\Lambda)^{-1}(P^{T}_{-1}(t),P^{T}_{-2}(t),\cdots, P^{T}_{-(m-1)}(t),P^{T}_{m}(t))^{T} \\
&=C(zI_{nm}-\Lambda)^{-1}U(t)B=C(zI_{nm}-\Lambda)^{-1}S^{-1}_{U}(t)Y_{U}(t)B \\
&=C(t)(zI_{nm}-H(t))^{-1} B(t)=F(z,t),
\end{align*}
where $H(t)=S_{U}(t)\Lambda S^{-1}_{U}(t)$, and both $C(t)$, and $B(t)$ as in the proposition \ref{alerta}. Therefore the condition is sufficient.
The necessity is clear using again the previous remark.
\end{proof}

We continue with our study of the properties of the family of linear system (\ref{SDL}).

\begin{proposition} The parametric family of linear dynamical system (\ref{SDL}), with $C(t)$, $B(t)$ and $H(t)$  defined as in the
proposition \ref{alerta} such that $U(0)=I_{nm}$, is controllable, if we assume that $|M_{-m}|\neq 0$ where $$B=(M^{T}_{-1},M^{T}_{-2},\cdots,M^{T}_{-(m-1)},M^{T}_{-m})^{T},$$ and observable in any case.
\end{proposition}
\begin{proof}In fact, we have
\begin{align*}
&rank(B(t)\,\,\,H(t)B(t)\,\,\cdots\cdots \,\,H^{{m-1}}(t)B(t)) \\
&=rank(Y_{U}(t)B\,\,\,\,S_{U} \Lambda S_{U}^{-1}(t)Y_{U}(t)B\,\,\cdots\cdots\,\,S_{U}(t)\Lambda ^{{m-1}}S_{U}^{-1}(t)Y_{U}(t)B) \\
&=rank(S_{U}(t)(U(t)B\,\,\,\,\Lambda U(t)B\,\,\,\cdots\cdots\,\,\,\Lambda^{{m-1}}U(t)B).
\end{align*}

Now under our assumption $\Lambda$ and $U(t)$ commute (see the proof of lemma \ref{otro llamado}). Hence
\begin{equation*}
rank(B(t)\,\,\,H(t)B(t)\,\,\cdots\cdots \,\,H^{{m-1}}(t)B(t))=rank(Y_{U}B\,\,\,Y_{U}\Lambda B\,\,\,\cdots\cdots\,\,\,Y_{U}\Lambda^{m-1}B).
\end{equation*}

Thus, (\ref{SDL}) is controllable if and only if $rank(Y_{U}B\,\,\,Y_{U}\Lambda B\,\,\,\cdots\cdots\,\,\,Y_{U}\Lambda^{m-1}B)=mn$. But quickly one sees that
\begin{equation*}
|Y_{U}B\,\,\,\,Y_{U}\Lambda B\,\,\,\cdots\cdots\,\,\,Y_{U}\Lambda^{m-1}B|=|Y_{U}||B\,\,\,\Lambda B\,\,\,\cdots\cdots\,\,\,\Lambda^{m-1}B|\neq 0,
\end{equation*}
this implies the claimed result. \\

We shall show the observably. Notice that
\begin{align*}
&rank(D(t)\,\,\,\,H^{T}(t)D(t)\,\,\,\cdots\cdots \,\,\,(H^{T})^{m-1}(t)D(t)) \\
&=rank((S_{U}^{T}(t))^{-1}D\,\,\,\,\,(S_{U}^{T}(t))^{-1}\Lambda^{T}S_{U}^{T}(t)(S_{U}^{T}(t))^{-1}D\,\,\,\cdots\cdots\,\,\,(S_{U}^{T}(t))^{-1}
(\Lambda^{T})^{m-1}S_{U}^{T}(t)(S_{U}^{T}(t))^{-1}D) \\
&=rank((S_{U}^{T}(t))^{-1}(D\,\,\,\,\Lambda ^{T}D\,\,\,\cdots\cdots\,\,\,(\Lambda^{T})^{m-1}D))=mn,
\end{align*}
this last is because
\begin{equation*}
|(S_{U}^{T}(t))^{-1}(D\,\,\,\,\Lambda ^{T}D\,\,\,\cdots\cdots\,\,\,(\Lambda^{T})^{m-1}D)|=|(S_{U}^{T}(t))^{-1}|\,\,|D\,\,\,\,\Lambda ^{T}D\,\,\,\cdots\cdots\,\,\,(\Lambda^{T})^{m-1}D|\neq 0,
\end{equation*}
hence, the parametric family of linear dynamical system (\ref{SDL}) is observable.
\end{proof} \\

Consider a transfer function $F(z,t)=C(zI_{nm}-\Lambda)^{-1}U(t)B$ of (\ref{SDL}) such that $U(0)=I_{nm}$, that is, let $F(z,t)$ be a transfer
function for which $U(t)=Y_{E}(t_{1},\cdots,t_{m-1})$, then
$$F(z,0)=C(zI_{nm}-\Lambda)^{-1}B=F_{0}(z),$$
and we can characterize the flow of $F(z,t)$. When calculating the derivative with respect to $t_{k}$ of  $F(z,t)$, we obtain
\begin{equation*}
\frac{\partial F(z,t)}{\partial t_{k}}=C(zI_{nm}-\Lambda)^{-1}\frac{\partial U(t)}{\partial t_{k}}B=
C(zI_{nm}-\Lambda)^{-1}\Lambda^{k}U(t)B=C(zI_{nm}-\Lambda)^{-1}U(t)\Lambda^{k}B.
\end{equation*}

\section{Conclusions}

In this paper, we introduced and studied an integrable system (hierarchy) called by us, \textbf{the block matrices version of the finite discrete KP hierarchy}. In addition, we introduced a group factorization for equation system, necessary to connect the control theory of linear dynamical systems with this integrable system. Thus, we established a correspondence between the solutions of the hierarchy with a parametric linear system. We see that the linear system defined by means of the simplest solution of the integrable system is controllable and observable. Then, because of this fact, it is possible to verify that any solution of the integrable hierarchy, obtained by the dressing method of the simplest solution, defines a parametric linear system, which is also controllable and observable. Finally, we studied the transfer function family corresponding to parametric linear systems whose coefficients are block matrices. Thus, these transfer functions constitute Laurent polynomials whose coefficients are square matrices.

\section*{Acknowledgment}

Nancy L\'{o}pez thanks CIMAT for its hospitality during her visit to the Center between March $19$ and $23$ of 2018, and also thanks the financial
support through CONACYT project $45886$. Ra\'{u}l Felipe-Sosa thanks the support of the Mexico Secretary of Education (SEP) and the hospitality of
the School of Physical-Mathematical Sciences of the BUAP, during his postdoctoral stay, from July $2017$ to July $2018$, period in which part of this
work was completed. The last named author was supported by CONACYT grant $45886$.

We would like to thank the anonymous reviewers who with their comments helped improve the manuscript.


\begin{thebibliography}{40}

\bibitem{alvarez} C. \'{A}lvarez-Fern\'{a}ndez, G. Ariznabarrete, J. C. Garc\'{\i}a-Ardila, M. Ma\~{n}as and F. Marcell\'{a}n,
\emph{Christoffel Transformations for Matrix Orthogonal Polynomials in the Real Line and the non-Abelian 2D Toda Lattice Hierarchy.}
Int. Math. Res. Not. IMRN \textbf{2017} (2016), 1285–1341.
\bibitem{ariznabarrete} G. Ariznabarrete and M. Ma\~{n}as, \emph{Matrix orthogonal Laurent polynomials on the unit circle and Toda type integrable systems.} Adv. Math. \textbf{264} (2014), 396-463.
\bibitem{brockett} R. W. Brockett and P. S. Krishnaprasad, \emph{A scaling theory for linear systems.} IEEE Trans. Automat. Control
\textbf{25} (1980), 197-207.
\bibitem{camara} M. C. C\^{a}mara, A. F. dos Santos and P. F. dos Santos, \emph{Lax equations, factorizacion and Riemann-Hilbert problems.} Port. Math. \textbf{64} (2007), 509-533.
\bibitem{CZ} R. F. Curtain a,nd H.J. Zwart, \emph{An introduction to infinite dimensional systems theory.} Texts in Applied Mathematics \textbf{21}, Springer-Verlag, New York, 1995.
\bibitem{felipe1} R. Felipe and F. Ongay, \emph{Algebraic aspects of the discrete KP hierarchy.} Linear Algebra Appl. \textbf{338} (2001), 1-17.
\bibitem{felipe2} R. Felipe and N. L\'{o}pez-Reyes, \emph{The finite discrete KP hierarchy and the rational functions.}
Discret Dyn Nat Soc (2008) Article ID 792632, 10 pages doi:10.1155/2008/792632.
\bibitem{fenga} Lian-Li Fenga, Shou-Fu Tiana, Xiu-Bin Wang, and Tian-Tian Zhang,  \emph{Rogue waves, homoclinic breather waves and soliton waves
for the $(2+1)$-dimensional $B$-type Kadomtsev-Petviashvili equation.} Appl. Math. Lett. \textbf{65} (2017), 90-97.
\bibitem{kharif} C. Kharif, E. Pelinovsky, and A. Slunyaev \emph{Rogue Waves in the Ocean.} Advances in Geophysical and Environmental Mechanics
and Mathematics. Springer-Verlag Berlin Heidelberg $2009$.
\bibitem{nancy} N. L\'{o}pez-Reyes, R. Felipe and T. Castro-Polo, \emph{The discrete KP hierarchy and the negative power series on the
complex plane.} Comp. Appl. Math. \textbf{32} (2013), 483-493.
\bibitem{nancy1} N. L\'{o}pez-Reyes and L. E. Benítez Babilonia, \emph{A discrete hierarchy of double bracket equations and a class of negative power series.} Math. Control and Related Fields \textbf{7} (2017), 41-52.
\bibitem{nakamura} Y. Nakamura, \emph{Geometry of rational functions and nonlinear integrable systems.} Siam. J. Math. Anal.
\textbf{22} (1991), 1744-1754.
\bibitem{onorato} M. Onorato, S. Residori, and F, Baronio, Editors  \emph{Rogue and Shock Waves in Nonlinear Dispersive Media.} The Lecture Notes in
Physics $926$. Springer, $2016$.
\bibitem{peng} Wei-Qi Peng, Shou-Fu Tian, and Tian-Tian Zhang, \emph{Analysis on lump, lumpoff and rogue waves with predictability to the
$(2+1)$-dimensional $B$-type Kadomtsev-Petviashvili equation.} Phys. Lett. A \textbf{382} (2018), 2701-2708.
\bibitem{qin} Chun-Yan Qin, Shou-Fu Tian, Xiu-Bin Wang, Tian-Tian Zhang,  and Jin Li, \emph{Rogue waves, bright–dark solitons and traveling wave
solutions of the $(3+1)$-dimensional generalized Kadomtsev-Petviashvili equation.} Comput. Math. Appl. \textbf{75} (2018), 4221-4231.
\bibitem{schwarz-zaks} B. Schwarz and A. Zaks, \emph{Geometry of matrix differential systems.} J. Math. Anal. Appl. \textbf{112} (1985), 165-177.
\bibitem{semenov} M. Semenon-Tian-Shansky, \emph{Integrable Systems and Factorization Problems.} Factorization and Integrable Systems,
Operator Theory Advances and Applications. $2000$, Vol. $141$.
\bibitem{tian} Shou-Fu Tian, and Hong-Qing Zhang, \emph{On the integrability of a generalized variable-coefficient Kadomtsev-Petviashvili equation.}
J. Phys. A: Math. Theor. \textbf{45} (2012), 055-203.
\bibitem{tu} Jian-Min Tu, Shou-Fu Tian, Mei-Juan Xu, Pan-Li Ma, and Tian-Tian Zhang,  \emph{On periodic wave solutions with asymptotic behaviors
to a $(3+1)$-dimensional generalized $B$-type Kadomtsev-Petviashvili equation in fluid dynamics.} Comput. Math. Appl. \textbf{72} (2016), 2486-2504.
\bibitem{wang} Xiu-Bin Wang, Shou-Fu Tian, Chun-Yan Qin, and Tian-Tian Zhang,  \emph{Characteristics of the solitary waves and rogue waves with
interaction phenomena in a generalized $(3+1)$-dimensional Kadomtsev-Petviashvili equation.} Appl. Math. Lett. \textbf{72} (2017), 58-64.
\bibitem{wang1} Xiu-Bin Wang, Shou-Fu Tian, Hui Yan, and Tian Tian Zhang, \emph{On the solitary waves, breather waves and rogue waves to a
generalized $(3+1)$-dimensional Kadomtsev-Petviashvili equation.} Comput. Math. Appl. \textbf{74} (2017), 556-563.
\bibitem{wang2} Xiu-Bin Wang, Shou-Fu Tian, Lian-Li Feng, Hui Yan, and Tian-Tian Zhang,  \emph{Quasiperiodic waves, solitary waves and asymptotic
properties for a generalized $(3+1)$-dimensional variable-coefficient $B$-type Kadomtsev-Petviashvili equation.} Nonlinear Dyn \textbf{88}
(2017), 2265-2279.
\bibitem{yan} Xue-Wei Yan, Shou-Fu Tian, Min-Jie Dong, and Li Zou, \emph{B\"{a}cklund transformation, rogue wave solutions and
interaction phenomena for a $(3+1)$-dimensional $B$-type Kadomtsev-Petviashvili–Boussinesq equation.} Nonlinear Dyn \textbf{92} (2018), 709-720.
\bibitem{zhang} F. Zhang, \emph{The Schur complement and its applications.} Springer, $2005$. \\
\end{thebibliography}
\end{document}